\documentclass[conference,letterpaper]{IEEEtran}
\IEEEoverridecommandlockouts
\usepackage{multirow}   

\usepackage[utf8]{inputenc} 
\usepackage[T1]{fontenc}
\usepackage{url}
\usepackage{ifthen}
\usepackage{cite}
\usepackage{comment}
\usepackage[cmex10]{amsmath}
\usepackage{mathtools}
\usepackage{pgfplots}
\pgfplotsset{compat=newest}

\usepackage{amsthm}
\usepackage{graphicx}
\usepackage{algorithm, algorithmic}
\usepackage{subcaption}
\usepackage{amsfonts}
\usepackage{thmtools}
\usepackage{booktabs}
\usepackage{float}
\usepackage{amssymb}
\usepackage{amsmath}
\usepackage{amsthm}
\usepackage{xspace}
\usepackage[hidelinks]{hyperref}
\usepackage[capitalize]{cleveref}
\usepackage[inline]{enumitem}
\usepackage{cite}
\usepackage{siunitx}
\usepackage{dsfont}
\usepackage{balance}
\usepackage{bm}

\newcounter{eqstep}

\usepackage{tikz}
\usetikzlibrary{arrows.meta, positioning, shapes.geometric, bending, calc, arrows.meta, fit}
\newcommand{\sX}{\mathcal{X}}
 \newcommand{\utag}[2]{\mathop{#2}\limits^{\text{(#1)}}}
\newcommand{\uref}[1]{(#1)}
\newcommand{\pb}[1]{p\left(#1\right)}
\declaretheorem[name=Definition]{definition}
\newtheorem{thrm}{Theorem}
\newtheorem{prop}{Proposition}
\newtheorem{lem}{Lemma}

\newtheorem{exmp}{Example} 
\DeclarePairedDelimiter\parens{\lparen}{\rparen}  

\DeclarePairedDelimiter\norm{\lVert}{\rVert}

\DeclarePairedDelimiter\bracks{\lbrack}{\rbrack}

\newcommand{\set}[1]{\bracks*{#1}}

\newcommand{\E}[1]{\mathbb{E}\bracks*{#1}}

\usepackage{xcolor}
\usepackage[normalem]{ulem}

\renewcommand{\P}[1]{\mathbb{P}\parens*{#1}}
\newcommand{\independent}{\perp\!\!\!\perp}

\begin{document}
 \title{Perfect Privacy and Strong Stationary Times for Markovian Sources } 


\author{
\IEEEauthorblockN{
Fangwei Ye\IEEEauthorrefmark{1},
Zonghong Liu\IEEEauthorrefmark{2},
Parimal Parag\IEEEauthorrefmark{3},
Salim El Rouayheb\IEEEauthorrefmark{2}
}
\IEEEauthorblockA{\IEEEauthorrefmark{1}
College of Computer Science and Technology,  Nanjing University of Aeronautics and Astronautics, Nanjing, China \\ Email: fangweiye@nuaa.edu.cn
}
\IEEEauthorblockA{\IEEEauthorrefmark{2}
Department of Electrical and Computer Engineering,
Rutgers University,
New Brunswick, NJ, USA\\ Email: \{zonghong.liu, salim.elrouayheb\}@rutgers.edu
}

\IEEEauthorblockA{\IEEEauthorrefmark{3}
Department of Electrical Communication Engineering, Indian Institute of Science, 
Bengaluru, Karnataka, India\\
Email: parimal@iisc.ac.in
}
                    }


\maketitle
\begin{abstract}
We consider the problem of sharing correlated data under a perfect information-theoretic privacy constraint. We focus on redaction (erasure) mechanisms, in which data are either withheld or released unchanged, and measure utility by the average cardinality of the released set, equivalently, the expected Hamming distortion. Assuming the data are generated by a finite time-homogeneous Markov chain, we study the protection of the initial state while maximizing the amount of shared data. We establish a connection between perfect privacy and window-based redaction schemes, showing that erasing data up to a strong stationary time preserves privacy under suitable conditions. We further study an optimal sequential redaction mechanism and prove that it admits an equivalent window interpretation. Interestingly, we show that both mechanisms achieve the optimal distortion while redacting only a constant average number of data points, independent of the data length~$N$.

\end{abstract}

\section{Introduction}
 
 A data owner holds a dataset $\{X_n\}_{n=0}^N$ consisting of possibly correlated data points and seeks to maximize the information disclosed to a third party (e.g., a data broker), subject to a perfect privacy constraint. Let $\mathcal{S} \subseteq \{0,1,\dots, N\}$ denote the indices of data points that must remain private, and let $X_{\cal S} = \{X_i: i \in \cal S\}$. Because correlations among data points may cause indirect leakage of private information, the released dataset $Y$ is obtained by selectively redacting the original data so as to ensure, if possible, that $Y$ reveals no information about $X_{\cal S}$, i.e., $X_{\cal S}$ and $Y$ are statistically independent, and
\begin{equation}
I(X_{\cal S}; Y) = 0.
\end{equation}

This setup is motivated by a range of data-sharing applications in which a third party wants to acquire a faithful (unmodified) subset of the data, rather than a perturbed or noisy version. Accordingly, we focus on \emph{redaction- or erasure-based data sharing mechanisms}, in which certain data points are withheld while the remaining data are released unchanged.

For example, the data owner may be a social network in which data points correspond to individuals, and a subset of users has opted to remain private, as required by regulations such as the General Data Protection Regulation (GDPR) and the California Consumer Privacy Act (CCPA). In this setting, data may be shared with a broker or used for targeted advertising, subject to strict privacy guarantees. In another scenario, data points may represent features of individuals or groups, where certain attributes (e.g., salary or clinical data) may be disclosed, while sensitive attributes such as race or gender must remain private.

In the case where the data points $X_i$ are independent, the problem admits a straightforward and optimal solution: the data owner can share all $X_i$ such that $i \notin {\cal S}$. The setting becomes more interesting when the data points are correlated. To make the model more specific and gain analytical insight, we focus on the case where $\{X_n\}_{n=0}^N$ forms a time-homogeneous Markov chain,
\begin{equation}
X_0 \to X_1 \to \cdots \to X_N,\label{mc}
\end{equation}
and the data owner seeks to protect, for example, $X_0$, i.e., ${\cal S} = \{0\}$. In this setting, regardless of how large $N$ is, every data point $X_i$ with $i \neq 0$ is generally correlated with $X_0$, and releasing it would therefore leak information about the private data. At first glance, this suggests that the data owner may be unable to share any data without violating the perfect privacy constraint.


In this paper, we show that this intuition does not hold in general. Specifically, we present two data-sharing schemes that achieve perfect privacy while, on average, redacting only a constant number of data points, independent of $N$. Our main idea is a \textit{data-dependent} randomized redaction window strategy in which all data points within a randomly sized window, whose length depends on the realized data, are withheld, while all data outside this window are released unchanged. Towards that end, we establish a novel connection to the theory of strong stationary times for Markov chains~\cite{levin2017markov}, which characterize stopping times that yield exact samples from the stationary distribution and which were originally motivated by the classical card shuffling problem~\cite{aldous1986shuffling}. 


\subsection{Related Works}

The tension between correlation and privacy has been studied in the context of differential privacy (DP). The standard DP formulation is distribution-agnostic and was originally developed for i.i.d. data (see, e.g., \cite{song2017pufferfish}).  Numerous DP variants have been proposed to account for correlations in the data \cite{kifer2011no, cao2017quantifying, he2014blowfish, zhu2014correlated, chen2014correlated, yang2015bayesian, liu2016dependence, gedik2007protecting, xiao2015protecting, song2017pufferfish}. Addressing correlation is more intuitive to handle under DP, since the framework permits a controlled amount of information leakage. For example, in the Markov chain setting considered here, correlations decay over time, and \(\epsilon\)-DP can be achieved by releasing all data beyond a fixed time threshold that depends only on the data distribution and the privacy budget \(\epsilon\)\footnote{This scheme is actually not optimal as shown in \cite{massny2025between}.}. Our work differs from this body of literature in two fundamental ways: first, privacy is required only for a subset of users or data points; and second, those that do require privacy insist on perfect privacy, i.e., zero information leakage.

Our problem can be viewed as a variant of the \emph{privacy funnel} (PF) \cite{makhdoumi2014information, kung2018compressive}, which itself is the dual of the \emph{information bottleneck} (IB) compression problem \cite{tishby2000information}. In the PF framework, given correlated private data \(X\) and public data \(Y\), the objective is to design a released variable \(U\),  satisfying the Markov constraint \(X \rightarrow Y \rightarrow U\), that maximizes the utility characterized there by the mutual information \(I(U; Y)\) (instead of the cardinality of released dataset as in here),  subject to a relaxed privacy constraint \(I(U; X) \le \epsilon\). Subsequent works have characterized the fundamental limits of the privacy funnel and established achievability results under general joint distributions \cite{calmon2015fundamental, zamani2023new, zamani2024statistical, shkel2020secrecy, rassouli2023information}.

The scenario in which data has heterogeneous privacy requirements—where only a subset of requests must remain private—was first explored in earlier work, including studies motivated by genomic data privacy \cite{ye2022mechanisms}. The genomic privacy problem concerns the release of genomic data while ensuring perfect privacy for parts of it that may be correlated with predisposition to certain diseases \cite{ye2022mechanisms}. This notion of selective privacy was first formalized in the context of \emph{On--Off privacy} \cite{naim2019off, ye2019preserving, ye2021off, ye2021onoff}. In this setting, which arises in private information retrieval, user requests may be correlated, and the challenge is to design queries that enable data retrieval while concealing the private requests.  



One of the main results of this paper is to establish a novel connection between the notion of perfect privacy and the classical concept of strong stationary times (SSTs), originally introduced in \cite{aldous1986shuffling, aldous1987strong}. SSTs have been extensively studied as a probabilistic tool for analyzing properties of Markov chains, such as mixing times \cite{diaconis1990strong, levin2017markov, aldous-fill-2014} and exact sampling algorithms \cite{lovasz1995efficient, fill1997interruptible}, among others. We show that, in certain cases, SSTs provide a natural framework for constructing data-release mechanisms that achieve perfect privacy for Markovian data models.

\subsection{Contributions}

We study the problem of data sharing under a perfect privacy constraint when the underlying data are correlated. 
We focus on the simplest nontrivial model of correlated data, namely a time-homogeneous Markov chain as in \eqref{mc}. In this model, we investigate whether it is possible to preserve the privacy of the initial state \(X_0\) while releasing a subset of the remaining data. For instance, if the Markov chain represents a random walk on a finite graph, the privacy constraint corresponds to hiding the starting node, as required, for example, to guarantee anonymity in a routing protocol. This setting serves as a natural first step toward more general problems involving the protection of an arbitrary subset \(S\) of data and richer models of data correlation. We make the following  contributions:
 
\noindent  \textit{Perfect privacy via SST-Mechanisms.}   
We propose a data-sharing scheme that redacts all data points within a random window $[0,\tau-1]$ and releases the remaining data $(X_\tau, X_{\tau+1}, \ldots, X_N)$ unchanged, where $\tau$ is chosen to be a strong stationary time (SST) of the Markov chain. Intuitively, a strong stationary time $\tau$ is a randomized stopping time at which the state of the Markov chain $X_\tau$ is exactly distributed according to its stationary distribution and is independent of $\tau$. We give a sufficient condition on the transition matrix $P$ under which this SST-mechanism achieves perfect privacy, and show that, it achieves optimal utility.

\noindent  \textit{Sequential Markov Redaction Mechanisms for General Chains.}  
Our SST-based approach  serves as a key inspiration for constructing our optimal Sequential Markov Redaction SMR-Mechanism. This mechanism   also operates by redacting data within a random window $[0,T-1]$, where $T$ is a refinement of the SST concept designed to ensure the privacy of $X_0$. This optimal scheme can be viewed as a reinterpretation of earlier work originally developed in the context of genomic privacy~\cite{ye2022mechanisms}. For this scheme, we provide a complete characterization of the achievable utility and show that it redacts, on average, only a constant number of data points, depending only on the transition matrix $P$.

\section{Problem Setting}
Let $\{X_i\}_{i=0}^N$ be the original data held by the data owner, $\{Y_i\}_{i=0}^N$ be an obfuscation of the data points $\{X_i\}_{i=0}^N$.
The privacy mechanism is determined by the conditional distribution $\P{ Y_{[0:N]} \mid X_{[0:N]}}$, illustrated by the following channel  
\begin{align}
    (X_0,...,X_N)\rightarrow\boxed{\P{Y_{[0:N]}\mid  X_{[0:N]}}}\rightarrow (Y_1,...,Y_N),\label{channel}
\end{align}
adhering to the perfect privacy constraint: 
\begin{equation}
\label{pp}
I(X_0; Y_0, Y_1, \ldots,Y_N) = 0.
\end{equation}
To achieve perfect privacy, we consider employing a redaction mechanism that selectively erases data points prior to sharing. As such, each shared data point, denoted by $Y_i$, consists of either the original data point  $X_i$  or an erasure symbol $\ast$, such that $Y_i \in \{X_i, \ast\}$.

To make the data correlation concrete, we consider the scenario where $\{X_i\}_{i=0}^N$ forms a time-homogeneous (discrete-time) Markov chain 
\begin{equation}
    X_0 \to X_1 \to \cdots \to X_N.\label{chain}
\end{equation}
Let $P$ denote the transition matrix of the underlying Markov chain, where $P(x,y)=\P{X_{1}=y\mid X_0=x}$, then the $t$-step transition probability $\P{X_t=y\mid X_0=x}$ is given by the $(x,y)$ element of $t$-th power of $P$, i.e., $P^t(x,y)$. In this work, we assume all Markov chains are defined on a finite state space $\sX$, and are aperiodic and irreducible. 

We measure the utility of the redaction mechanism by the amount of un-redacted data. 
Formally, the utility is the expected Hamming distance between the original and the released data, i.e.,  $$\E{d_H(X_{[0:N]}, Y_{[0:N]})}\triangleq\E{\big| \{t: Y_t=\ast,\,0\leq t\leq N\}\big|},$$  where the expectation is taken over the distribution of the data $X_{[0:N]}$ and the randomness of the mechanism. 

\section{Background on  Strong Stationary Times}

 We provide a brief overview of strong stationary times, following the standard exposition in \cite[Chapter~6]{levin2017markov}. Let $\{X_t\}_{t \ge 0}$ be a Markov chain with a stationary distribution $\pi$.
\begin{definition}
A random time $\tau$ is called a \emph{strong stationary time} for $\{X_t\}_{t \ge 0}$ starting from the initial state $X_0 = x$ if
\begin{align} 
    \mathbb{P}(\tau = t,\; X_\tau = y\mid X_0=x)
=
\mathbb{P}(\tau = t\mid X_0=x)\,\pi(y),\label{eq:SST}
\end{align}
for all $t \ge 0$ and all $y \in \mathcal{X}$.
\end{definition}

\begin{exmp}{(Top-to-random shuffle \cite{diaconis1992analysis}).} \normalfont 
    Consider a deck of $n$ cards that is repeatedly shuffled by removing the top card and reinserting it uniformly at random into one of the $n$ possible positions in the deck.  
This shuffling procedure can be modeled as a Markov chain over the set of all $n!$ permutations. Let $X_t$ denote the permutation after $t$ shuffling steps. Define $\tau$ as the first time at which the card that was initially at the bottom of the deck appears at the top and is reinserted. At this random time, one can show that the distribution of the deck configuration $X_\tau$ is uniformly distributed over all $n!$ permutations and is independent of  
the stopping time $\tau$ \cite[Chapter 6]{levin2017markov}. Consequently, $\tau$ is a strong stationary time. 
 \end{exmp}

A strong stationary time ensures that the Markov chain has reached its stationary distribution irrespective of its initial state $X_0$; hence, the potential connection to privacy. Therefore,  an implied mechanism to keep $X_0$ private would be to redact all data up to time $\tau-1$ and release the remaining data $(X_\tau, X_{\tau+1}, \ldots, X_N)$. While this guarantees privacy in the previous example, in general, this is not true. The reason is that privacy requires, in addition to the SST condition in \eqref{eq:SST}, that the stopping time $\tau$ itself be independent of the initial state. Equivalently,
\[
\mathbb{P}(\tau=t\mid X_0=x)=\mathbb{P}(\tau=t), 
\]
a condition that is not implied by the SST definition.


 
\begin{exmp}  \normalfont
    Consider the following Markov chain:
\begin{align*}
    \mathcal X=\{0,1\},\quad 
    P=\begin{pmatrix}
        1/2 & 1/2\\
        0 & 1
    \end{pmatrix},\quad 
    \pi=(0,1).
\end{align*}
One can construct a strong stationary time (SST) as follows.  
If $X_0=0$, define
\[
\tau=\min\{t\geq 0 : X_t=1\},
\]
so that $\tau\sim\operatorname{Geom}(1/2)$. Whereas, if $X_0=1$, set $\tau=1$ with probability one. Clearly, when $\tau=1$, it is more likely that $X_0=1$; that is, $\tau$ leaks information about $X_0$. On the other hand, consider an alternative random time $\tau'$. When $X_0=0$, define $\tau'$ in the same way as $\tau$. However, when $X_0=1$, draw $\tau'\sim\operatorname{Geom}(1/2)$ independently of the Markov chain. One can check that $\tau'$ is also an SST whose distribution is independent of $X_0$, and thus the pair $(X_{\tau'},\tau')$ preserves the privacy of $X_0$.
 \end{exmp}



  The above example also shows that a given Markov chain may admit multiple strong stationary times. 
  Aldous and Diaconis \cite{aldous1987strong} give a general construction of an \emph{optimal} SST $\tau$ that minimizes the expectation $\mathbb{E}[\tau]$ for a given initial state $x$. This construction helps identify conditions under which an SST can preserve privacy. See Section~\ref{sec:sst} for details.






\section{Main Results}

We present two redaction-based mechanisms for sharing Markovian data under perfect privacy for $X_0$ and characterize their utility:
\begin{enumerate}
    \item An SST-based mechanism, which uses SST as a release time and is private for \textit{transitively invariant} chains.
    \item A sequential Markov redaction (SMR) mechanism, which achieves perfect privacy for any finite Markov chain.
\end{enumerate}
Both mechanisms can be viewed as randomized, data-dependent window schemes that erase a prefix of the trajectory and release the remaining samples unchanged.

\subsection{SST-mechanism for transitively invariant chain}\label{sec:sst}
 
We start by introducing the construction of optimal SST given in \cite{aldous1987strong}. In this section, we assume all Markov chains under consideration are \emph{time-homogeneous, finite, aperiodic,} and \emph{irreducible}. We denote $P^t(x,y)=\P{X_t=y\mid X_0=x}.$
%

\vspace{.5mm}
\noindent\textbf{Construction of optimal SST.}    
Consider a time-homogeneous Markov chain $P$ on a finite state space $\sX$ starting from $x_0$. Let
\begin{align*}
    a_t^{x_0} = \min_{x\in\sX}\frac{P^t(x_0,x)}{\pi(x)},
\end{align*} 
and $\{U_t\}_{t=1}^N$ be a sequence of i.i.d. random variables uniformly distributed over the interval $(0,1)$  
 that is independent of $\{X_t\}_{t=0}^N$. Then the following random time \cite[Sec. 6.7]{levin2017markov} 
\begin{align}
    \tau = \min\set{t\geq 1\mid U_t\leq\frac{a_t^{x_0}-a_{t-1}^{x_0}}{{P^t(x_0,X_t)}/{\pi(X_t)}-a_{t-1}^{x_0}}},\label{sst}
\end{align}
is an SST   of $\{X_t\}_{t=0}^N$ and satisfies:  
\begin{align}
   \mathbb{P}(\tau = t, X_\tau = y\mid X_0=x_0)=\pi(y)(a_t^{x_0}-a_{t-1}^{x_0}).\label{joint}
\end{align}
We have naturally, from the definition of SST in \eqref{eq:SST}: 
\begin{align}
    \mathbb{P}(\tau = t\mid X_0=x_0)&= a_t^{x_0}-a_{t-1}^{x_0}.\label{tau}
\end{align}

From \eqref{tau}, we observe that a sufficient condition for $\tau$ to be independent of $X_0$ is that $a_t^{x_0} - a_{t-1}^{x_0}$ is independent of $X_0$ for all $t$. This condition is implied if $a_t^{x_0} = \min_x P^t(x_0,x)/\pi(x)$ is independent of $x_0$ for all $t$, that is, if $a_t^{x_0}$ is indistinguishable across all initial states. A sufficient, though not necessary, condition to formalize such indistinguishability is symmetry induced by a transitive group action on the state space, which motivates the following definition, which has been studied in the task of random walk on groups \cite{diaconis1988group}.



 

\begin{definition}
    Let $\sX$ be a finite state space and $P$ a Markov transition matrix on $\sX$.
If there exists a group $G$ acting transitively on $\sX$, i.e., for any $x,y\in\sX$, there exists $g\in G$ satisfies $x=gy$, and the transition matrix $P$ satisfies
\[
P(gx,gy)=P(x,y), \qquad \forall g\in G,\ \forall x,y\in\sX,
\]
we call the Markov chain $P$ a transitively invariant chain. 
\end{definition}
Examples of transitively invariant Markov chains are those that can be viewed as a simple random walk on Cayley graphs (e.g., $d$-dimensional hypercubes, in which case the group $G$ can be taken to be the set of all cyclic shifts of the $d$ coordinates). From the perspective of transition matrices, any circulant transition matrix and its row permutations are transitively invariant. 

We are now ready to state our first result. Let $\tau$ be an optimal SST according to \eqref{sst}. We define the \textbf{SST-Mechanism} to be a mechanism that redacts data up to $\tau$ and releases the rest, i.e.,
\begin{align}
        Y_t =\begin{cases}
            \ast & \textit{ if } t<\tau,\\
            X_t & \textit{ otherwise},
        \end{cases}\label{sstscheme}
    \end{align}
where $\tau$ is constructed by \eqref{sst}.

\begin{thrm}[SST-Mechanism]\label{thm:sst}
    If $\{X_t\}_{t=0}^N$ forms a transitively invariant Markov chain and the initial data point $X_0$ is to be  protected,
    the SST-mechanism
    satisfies the perfect privacy constraint in \eqref{pp}.
       Moreover, it achieves the minimum distortion given by:
        \begin{align} 
    \E{d_H(X_{[0:N]},Y_{[0:N]})}=&N-\sum_{t= 0}^N |\sX| \min _{x \in \sX} P^t(x_0, x) , \label{sstdistortion}
        \end{align}
        which is independent of $x_0$.
         \end{thrm}

\subsection{SMR-mechanism for   general chains}

We observe that the privacy constraint \eqref{pp} can be expanded by the chain rule as
\[I(X_0;Y_0,Y_1,\ldots,Y_N) = \sum_{t = 0}^{N} I(X_0;Y_t|Y_{[0:t-1]}).\]
This observation gives rise to a sequential mechanism for generating $Y_t$ according to previously generated $Y_0, \ldots, Y_{t-1}$ such that $I(X_0;Y_t|Y_{[0:t-1]}) = 0$ is satisfied for all $t$. This motivates the following construction of the Sequential Markov Redaction \textbf{SMR-Mechanism}, which redacts
 the data one by one and designs the redaction probability based on previously generated $Y_{[0:t-1]}$, $X_0$, and $X_t$. More specifically, let
\begin{align}
    q_t = \frac{\min_{u \in \mathcal{X}} \P{X_t = x_{t}|X_{0}=u,Y_{[0:t-1]}=y_{[0:t-1]}}}{\P{X_t = x_{t}|X_{0}=x_0,Y_{[0:t-1]}=y_{[0:t-1]}}}.\label{eq:eraseprob}
\end{align}
The SMR mechanism releases the data $Y_t$ according to
\begin{equation}
Y_t =
\begin{cases}
\ast, & \text{with probability } 1- q_t, \\
X_t, & \text{otherwise}.
\end{cases}
\label{eq:qdef}
\end{equation}



\begin{thrm}[SMR-Mechanism]
\label{thm}
If $\{X_t\}_{t=0}^{N}$ forms a Markov chain and the initial data point $X_0$ is to be protected, then the SMR-mechanism defined in \eqref{eq:qdef} satisfies the perfect privacy constraint in \eqref{pp}. Moreover, it achieves the optimal Hamming distortion given by
\begin{align}\label{SMRBound}
 \E{d_H(X_{[0:N]}, Y_{[0:N]})} =  
 \sum_{t=0}^{N} \left( 1 - \sum_{x \in \mathcal{X}} \min_{u \in \mathcal{X}} P^t(u,x) \right). 
\end{align} 
\end{thrm}

The SMR-mechanism is a revisit of the privacy mechanism in \cite{ye2022mechanisms} in the context of the haplotype hiding problem in genomics, where we showed that it is private and   optimal in the sense of minimizing the Hamming distortion. The proof can be adapted from Corollary~1 in \cite{ye2022mechanisms} and the details are omitted here. 

In light of the previous connection to SSTs, we will prove that the SMR-Mechanism \eqref{eq:qdef} can also be viewed as a window-based redaction scheme. It redacts all data points within the window starting from $X_0$ and releases the data points outside the window. We denote by $T$ the length of the window.%
\begin{thrm}[Window Interpretation]
\label{thm:window} The SMR-mechanism is given by
\begin{align}
        Y_t =\begin{cases}
            \ast & \textit{ if } t<T,\\
            X_t & \textit{ otherwise,}
        \end{cases} \label{smrscheme}
    \end{align}
where $T$ is an integer random variable, depending on $X_{[1,N]}$ but not on $X_0$, with marginal distribution 
\begin{equation}
\label{eq:winsize}
    \P{T = t} = \alpha_{t} - \alpha_{t-1},
\end{equation}
for $t = 1,2,\ldots,N+1$, 
where
\begin{equation}
    \alpha_t=\sum_{v \in \mathcal{X}} \min _{u \in \mathcal{X}} \P{X_t = v \mid X_0=u }, t = 0,1,\ldots,N,
\end{equation}
are parameters completely determined by the transition matrix $P$ of the given Markov chain. 
\end{thrm}


Equation~\eqref{SMRBound} in Theorem~\ref{thm} fully characterizes the optimal utility and coincides with the distortion expression in Theorem~\ref{thm:sst} for transitively invariant chains. However, this characterization is not transparent and offers limited intuition.  
Therefore, we present the following theorem, which provides a transparent upper bound on the distortion in terms of the spectral properties of $P$. More importantly, the upper bound implies that only a constant average number of data points (independent of $N$) will be redacted. 
\begin{thrm}\label{thm:upper}
For an ergodic chain $P$, the optimal Hamming distortion for a redaction-based scheme that achieves perfect privacy constraint \eqref{pp}, which protects $X_0$, satisfies 
    \begin{align*}
        \E{d_H(X_{[0:N]}, Y_{[0:N]})}\leq\frac{|\sX|}{2\sqrt{\pi_{\min} }(1-\sqrt{\lambda})},
    \end{align*}
where $\lambda$ is the second largest eigenvalue of the multiplicative reversiblization $P\hat{P}$, $\hat{P}$ is the time reversal of $P$ such that $\hat{P}(x,y)=\frac{\pi(y)P(y,x)}{\pi(x)}$, and $\pi_{\min}=\min_{x\in\sX}\pi(x)$. 
\end{thrm}
\begin{proof}
    See Appendix~\ref{proofofcoro}.
\end{proof}

\section{Proof of theorems}
\subsection{Proof of Theorem~\ref{thm:sst}}\label{SSTSec}  To prove the perfect privacy part of Theorem 1, we need Lemmas 1 and 2. Lemma 1 shows that it is sufficient to protect the first releasing value $X_\tau$ and the release time $\tau$, which follows from the strong Markov property and the definition of $\tau$ (proof in Appendix A). 
 
\begin{lem}\label{lem:equivalent}
For the SST-Mechanism defined in \eqref{sstscheme}, we have
 $$   I(X_0;Y_{[0:N]}) = 0  \Leftrightarrow I(X_0;\tau,X_\tau)=0    .$$ 
\end{lem}

\begin{lem}\label{lem:sstprivacy}
When $P$ is transitively invariant, the SST-Mechanism in Theorem~\ref{thm:sst} guarantees $I(X_0;\tau,X_\tau)=0$.
\end{lem}
\begin{proof}[Proof of Lemma~\ref{lem:sstprivacy}]
    It suffices to show 
    \begin{align}
        \P{X_\tau=x,\tau=t\mid X_0=x_0}=\P{X_\tau=x,\tau=t}.
    \end{align}
    We show in Appendix~\ref{appendix:SSTSec} that transitively invariant chains satisfy the following two properties due to their symmetries: (i). $\min_xP^t(x_0,x)$ doesn't depend on $x_0$ for any $t$; (ii). $\pi$ is uniform over $\sX$.
    From equation~\eqref{joint}, we have
    \begin{align}
        &\mathbb{P}(\tau = t, X_\tau = x\mid X_0=x_0)\notag\\
        \overset{(a)}{=}&\pi(x){|\sX|}\left(\min_xP^t(x_0,x)-\min_xP^{t-1}(x_0,x)\right)\label{dist}\\
        \overset{(b)}{=}&\P{X_\tau=x,\tau=t},\notag
    \end{align}
    where (a) is due to the definition of $a_t^{x_0}$ and $\pi$ is uniform, and (b) is because $\min_xP^t(x_0,x)$ is independent of $x_0$.
\end{proof}
The utility in \eqref{sstdistortion} is proved in Lemma 3.
\begin{lem}
    For the SST-Mechanism defined in \eqref{sstscheme}, we have
    \begin{align*}
          \E{d_H(X_{[0:N]},Y_{[0:N]})}=&N-\sum_{t= 0}^N|\sX| \min _{x \in \sX} P^t(x_0, x).
    \end{align*}
\end{lem}
\begin{proof}
 Note that $\E{d_H(X_{[0:N]},Y_{[0:N]})}=\sum_{t=0}^N\P{Y_t=\ast}$.
From \eqref{tau} and \eqref{dist} and the perfect privacy, we have,
  \begin{align*}
       \P{Y_t=\ast} = 1 -\P{\tau\leq t}=1-|\sX|\min_x P^t(x_0,x).
    \end{align*}
    And the result follows from summing over $t$.
    It is worth noting that this distortion meets the optimal bound given in \cite{ye2022mechanisms}. 
\end{proof}

\subsection{Proof of Theorem~\ref{thm:window}}
Towards an understanding of the sequential redaction mechanism, we first note the following properties.
\begin{enumerate}[label=(\arabic*)]
    \item If $t = 0$, then the numerator in  \eqref{eq:eraseprob} becomes
    \begin{equation*}
        \min_{u \in \mathcal{X}} \P{X_0 = x_{0}|X_{0}=u} = 0,
    \end{equation*}
    by choosing the hypothetical value $u \neq x_0$, 
which yields
\begin{equation*}
\P{Y_t = x_t \mid X_{0}=x_{0}, X_t = x_t, Y_{[0:t-1]}=y_{[0:t-1]} }
    = 0.
\end{equation*}
This implies that $X_{t}$ is always redacted. Intuitively, the private data point $X_0$ must be masked before release. 
\item By inspecting the privacy mechanism in \eqref{eq:qdef}, we know that if $Y_{t-1} \neq \ast$ for some $t > 1$, then 
\begin{align*}
& \P{X_t = x_t|X_0 = x_0, Y_{[0:t-1]} = y_{[0:t-1]}} \\
 = &\P{X_t = x_t|X_0 = x_0, Y_{[0:t-1]} = y_{[0:t-1]},X_{t-1}=y_{t-1}}  \\
 = &\P{X_t = x_t|X_{t-1}=y_{t-1}}, 
\end{align*}
for any $x_t$ and $y_{[0:t-1]}$ by the Markov property. Note that $\P{X_t = x_t|X_{t-1}=y_{t-1}}$ does not depend on the hypothetical value $u$ of $X_0$. This implies that
\begin{align*}
& \P{Y_t = x_t \mid X_{0}=x_{0}, X_t = x_t, Y_{[0:t-1]}=y_{[0:t-1]} } \\
 = &\frac{\min_{u \in \mathcal{X}} \P{X_t = x_{t}|X_{0}=u,Y_{[0:t-1]}=y_{[0:t-1]}}}{\P{X_t = x_{t}|X_{0}=x_0,Y_{[0:t-1]}=y_{[0:t-1]}}} \\
 = &\frac{\P{X_t = x_t|X_{t-1}=y_{t-1}}}{\P{X_t = x_t|X_{t-1}=y_{t-1}}}  \\
 = &1,
\end{align*}
which means that if $Y_{t-1} \neq \ast$ then $Y_{t} \neq \ast$ with probability one.
\end{enumerate}

Given these two observations, we can conclude that the privacy mechanism redacts all data points within the window from the initial position and releases those outside the window.  
Therefore, the privacy mechanism can be viewed as determining the length of the redaction window, which is chosen stochastically. 

As discussed previously, we know that if $Y_{t-1} \neq \ast$, then $Y_{t} \neq \ast$, and thus the value of interest is the conditional probability $\P{Y_{t}= \ast \mid Y_{t-1}= \ast}$ induced by our privacy mechanism, which is given in the following lemma and the proof is deferred to Appendix~\ref{Appendix-A}.
\begin{lem}
\label{lem1}
The conditional probability $\P{Y_{t}= \ast \mid Y_{t-1}= \ast}$ induced by \eqref{eq:qdef} is given by
\begin{equation}\label{eq:lem}
    \P{Y_{t}= \ast \mid Y_{t-1}= \ast} = \frac{1-\alpha_t}{1-\alpha_{t-1}}.
\end{equation}
\end{lem}
From \eqref{eq:lem}, we can easily obtain that the probability distribution of the length of the redaction window is given by
\begin{equation*}
    \P{\tau = t} =  \P{Y_{t} \neq \ast, Y_{t-1} = \ast } = \alpha_{t} - \alpha_{t-1},
\end{equation*}
which finishes the proof.

\section{Discussion and conclusion}
For correlated data modeled by a finite time-homogeneous Markov chain, we characterize the optimal amount of data disclosed when protecting the initial state.  

A connection between privacy and strong stationary times is initially established. It turns out that, surprisingly, a sequential mechanism can be viewed as a relaxation of the SST-based mechanism. The difference between privacy and SST is that privacy merely requires its disclosure samples independent of private data, instead of a particular distribution, a stationary distribution in SST, which is a future direction for further exploration. 

\clearpage

\bibliographystyle{IEEEtran}
\bibliography{bib}
\appendices 
\onecolumn

\newpage

\section{Proof of Results in Section~\ref{SSTSec}}\label{appendix:SSTSec}
\begin{proof}[Proof of Lemma~\ref{lem:equivalent}]
    First, by the chain rule, we have
    \begin{align}
        I(X_0;Y_{[1:N]})&=I(X_0;Y_1,...,Y_\tau,Y_{\tau+1},...,Y_N)\\
        &\overset{(a)}{=}I(X_0;\tau, Y_\tau, Y_{\tau+1}, ..., Y_N)\\
        &\overset{(b)}{=}I(X_0;\tau, X_\tau, X_{\tau+1}, ..., X_N)\\
        &\overset{(c)}{=}I(X_0;\tau, X_\tau) + I(X_0; X_{\tau+1}, ..., X_N\mid \tau, X_\tau),
    \end{align}
    where (a) is due to $Y_0^{\tau-1}$ is a sequence of erasures with length $\tau$, and there is a one-to-one mapping between $Y_0^{\tau-1}$ and $\tau$, 
thus we have $H(Y_0^{\tau-1})=H(\tau)$ and $I(X_0;Y_0^{\tau-1})=I(X_0;\tau)$, (b) is due to $(Y_\tau,Y_{\tau+1},...,Y_N)=(X_\tau,X_{\tau+1},...,X_N)$, and (c) is the chain rule.
    Then it suffices to show $I(X_0; X_{\tau+1}, ..., X_N\mid \tau, X_\tau)=0$.

     In the construction of the release time $\tau$, the mechanism may employ an auxiliary sequence
of random variables $(U_t)_{t \ge 1}$, where $U_t \sim \mathrm{Unif}(0,1)$
are i.i.d. and independent of the entire chain $(X_t)_{t \ge 0}$.
Consider the \emph{extended filtration}
\begin{align}
    \mathcal G_t := \sigma\bigl(X_0, X_1, \ldots, X_t,\; U_1, U_2, \ldots, U_t\bigr),
\qquad t \ge 0.\label{extendfilter}
\end{align}
By construction, the stopping time $\tau$ satisfies
\[
\{\tau = t\} \in \mathcal G_t, \qquad \forall\, t \ge 0,
\]
and hence $\tau$ is a $\{\mathcal G_t\}$-stopping time.

Since the auxiliary randomness $(U_t)$ is independent of the chain
$(X_t)$ and does not affect its transition probability, the process
$(X_t)_{t \ge 0}$ remains a Markov chain with respect to the filtration
$\{\mathcal G_t\}$. In particular, the strong Markov property holds for
$(X_t)$ at the stopping time $\tau$ with respect to $\{\mathcal G_t\}$.
Specifically, for any $t \ge 0$, any $x \in \mathcal X$, and any $x_{t+1:n} \in \mathcal X^{n-t}$, we have
\begin{equation}
\mathbb P\!\left(
X_{t+1:n}=x_{t+1:n} \,\middle|\, \tau = t,\; X_\tau = x,\; \mathcal G_t
\right)
=
\mathbb P\!\left(
X_{t+1:n} =x_{t+1:n} \,\middle|\, \tau = t,\; X_\tau = x
\right).
\label{eq:strong_markov_extended}
\end{equation}
The right-hand side depends only on the current state $x$ and the time
horizon $n-t$, and is independent of both the initial state $X_0$ and the
auxiliary variables $(U_s)$.


\[
X_{\tau+1:n} \independent X_0 \mid (X_\tau=x, \tau=t).
\]
Since this factorization holds for all $(x,t)$, the conditional mutual
information satisfies
\[
I\!\left( X_{\tau+1}^{\,n} \,;\, X_0 \,\middle|\, Y_\tau, \tau \right) = 0 .
\]
\end{proof}

\begin{prop}
\label{Prop::TransitiveInvariant}
For a transitively invariant Markov chain, the following hold.

\begin{enumerate}
\item For every $t\ge 1$, the quantity
\[
\min_{y\in\sX} P^t(x,y)
\]
is independent of the initial state $x\in\sX$.

\item The transition matrix $P$ is doubly stochastic, and hence the uniform
distribution $\pi(y)=1/|\sX|$ is a stationary distribution of the chain.
\end{enumerate}
\end{prop}    
\begin{proof}[Proof of Proposition~\ref{Prop::TransitiveInvariant}]
We first show that for any $g\in G$, $gx$ is a bijection on $\sX$ to itself. Since $G$ is a group, there exists $e\in G$ such that $ex=x$ for any $x\in\sX$, and for any $g\in G$, there exists $g^{-1}\in G$ such that $g^{-1}g=gg^{-1}=e$. Thus, $g^{-1}gx=gg^{-1}x=x$, i.e., the inverse mapping exists, thus $g$ is a bijection.
\begin{enumerate}
    \item     We prove that $P^t$ holds $P^{t}(gx,gy)=P^t(x,y)$
by induction on $t$.
The case $t=1$ is the definition. Assume it holds for some $t\ge 1$.
Then for any $x,y$,
\[
P^{t+1}(gx,gy)=\sum_{z\in\Omega} P^{t}(gx,z)\,P(z,gy).
\]
Use the bijection $z=gz'$ (i.e., $z'=g^{-1}z$) to rewrite the sum as
\[
\sum_{z'\in\Omega} P^{t}(gx,gz')\,P(gz',gy).
\]
By the induction hypothesis and the invariance of $P$,
$P^{t}(gx,gz')=P^{t}(x,z')$ and $P(gz',gy)=P(z',y)$, hence
\[
P^{t+1}(gx,gy)=\sum_{z'\in\Omega} P^{t}(gx,gz')\,P(gz',gy)=\sum_{z'\in\Omega} P^{t}(x,z')\,P(z',y)=P^{t+1}(x,y).
\]
Now,
\[
\min_{y\in\sX} P^t(gx,y)
=
\min_{y\in\sX} P^t(gx,gy)
=\min_{y\in\sX} P^t(x,y),
\]
since $gx$ is a bijection from $\sX$ to $\sX$ for any $g\in G$ and the invariance $P^t(gx,gy)=P^t(x,y)$.
\item For any $g$ and $y$,
\[
\sum_x P(x,gy)=\sum_x P(gx,gy)=\sum_x P(x,y),
\]
where the change of variables uses that $x\mapsto gx$ is a bijection.
By transitivity, $\sum_x P(x,y)$ is constant in $y$; since each row of $P$
sums to $1$, this constant must be $1$ by double counting.
Hence, $P$ is doubly stochastic and the uniform distribution
$\pi(y)=1/|\sX|$ is stationary.
\end{enumerate}
\end{proof}

\section{Proof of Theorem~\ref{thm:upper}}\label{proofofcoro}

\begin{proof}
We have
    \begin{align}
        1-\sum_x\min_{x_0} P^t(x_0,x)
        &=\sum_x\left(\pi(x)-\min_{x_0}P^t(x_0,x)\right)\notag\\
        &\leq \sum_x\left|\pi(x)-\min_{x_0}P^t(x_0,x)\right|\notag\\
        &\leq |\sX|\max_x\left|\pi(x)-\min_{x_0}P^t(x_0,x)\right|.\label{corostep1}
    \end{align}
    For $x\in\sX$, since $\sX$ is finite, denote $z_x=\operatorname{argmin}_{x_0} P^t(x_0,x)$, then we have
    \begin{align*}
        \left|\pi(x)-\min_{x_0} P^t(x_0,x)\right|=\left|\pi(x)- P^t(z_x,x)\right|\leq\norm{\pi-P^t(z_x,\cdot)}_{TV},
    \end{align*}
    thus 
    \begin{align}
        \eqref{corostep1}&\leq |\sX|\max_x\norm{\pi-P^t(z_x,\cdot)}_{TV}\leq|\sX|\max_z\norm{\pi-P^t(z,\cdot)}_{TV},\label{corostep2}
    \end{align}
    the second inequality is because
    \begin{align*}
        \{z_x \mid z_x=\operatorname{argmin}_{x_0} P^t(x_0,x),x\in\sX\}\subset\sX.
    \end{align*}
  Since by assumption the chain is ergodic, by Theorem~2.1 in \cite{fill1991eigenvalue}, we have for any $v\in\sX$:
    \begin{align}
        \norm{\pi-P^t(x,\cdot)}_{TV}\leq\frac{\left(\sqrt{\lambda}\right)^t}{2\sqrt{\pi(x)}},\label{corostep3}
    \end{align}
    where $\lambda$ the second largest eigenvalue of matrix $P\hat{P}$, and $\hat P$ be the time reversal of $P$, i.e., $\hat{P}(x,y)=\frac{\pi(y)P(y,x)}{\pi(x)}$. 
    
    Using \eqref{corostep3}, we have by the definition of total variation distance
    \begin{align}
        1-\sum_x\min_{x_0} P^t(x_0,x)\leq |\sX|\max_z \frac{\left(\sqrt{\lambda}\right)^t}{2\sqrt{\pi(z)}} = \frac{|\sX|}{2\sqrt{\pi_{\min}}}\cdot\left(\sqrt{\lambda}\right)^t.
    \end{align}
    Thus we have 
    \begin{align*}
        \E{d_H(X_{[0:N]}, Y_{[0:N]})}&=\sum_{t=0}^N(1-\sum_x\min_{x_0} P^t(x_0,x))\\
        &\leq\sum_{t=0}^N\frac{|\sX|}{2\sqrt{\pi_{\min}}}\cdot\left(\sqrt{\lambda}\right)^t\\
        &\leq\sum_{t=0}^\infty\frac{|\sX|}{2\sqrt{\pi_{\min}}}\cdot\left(\sqrt{\lambda}\right)^t\\
        &=\frac{|\sX|}{2\sqrt{\pi_{\min}}(1-\sqrt{\lambda})}.
    \end{align*}
\end{proof}

 \section{Proof of Lemma~\ref{lem1}}
 \label{Appendix-A}
For the ease of notations, we will slightly abuse the notation by writing $\P{X=x}$ as $\pb{x}$ when the underlying random variable is clear in the context. Also, $[0:i-1]$ will be written by $[i-1]$.

Recall that the target is to prove 
\[\P{Y_i = \ast|Y_{i-1} = \ast} = \frac{1 - \alpha_i}{1 - \alpha_{i-1}}.\]
Invoking the property that $Y_i = \ast$ only if $Y_{i-1} = \ast$, which yields  
\[\P{Y_i = \ast|Y_{i-1} = \ast} = \P{Y_i = \ast|Y_{i-1} = Y_{i-2} = \cdots = Y_0 = \ast}.\]
Clearly, we are only interested in the case of $Y_i = Y_{i-1} = Y_{i-2} = \cdots = Y_0 = \ast$, and thus without specified otherwise, $y_{[i]}$ are assumed to be $\ast$ in the sequel. 

As such, the target probability $\P{Y_i = \ast|Y_{i-1} = \ast}$ can be written as $\pb{y_i|y_{[i-1]}}$. Due to the privacy constraint, we further have 
\[\pb{y_i|y_{[i-1]}} = \pb{y_i|x_0, y_{[i-1]}}, \forall x_0 \in \mathcal{X}.\]

Therefore, we can obtain from \eqref{eq:qdef} that 
\begin{align*}
p\left(y_i \mid x_0, y_{[i-1]}\right) & =\sum_{x_i} p\left(x_i \mid x_0, y_{[i-1]}\right) p\left(y_i \mid x_i, x_0, y_{[i-1]}\right) \\
& =1-\sum_{x_i} \min _{u \in \mathcal{X}} p\left(x_i \mid X_0=u, y_{[i-1]}\right),
\end{align*}
for any $x_{0} \in \mathcal{X}$. Let $c_i(j):=\min _{u \in \mathcal{X}} p\left(X_i=j \mid X_0=u, y_{[i-1]}\right)$ and $s_i=\sum_{j \in \mathcal{X}} c_i(j)$. Note that the difference between $s_i$ and $\alpha_i$ defined previously is that $s_i$ involves $y_{[i-1]}$ after the conditioning. 

We claim that $s_i$ equals to 
\begin{equation}
\label{eq:appendix-eq1}
    s_i =\frac{\sum_{x_i} \min _{u \in \mathcal{X}} p\left(x_i \mid X_0=u\right)-\sum_{x_{i-1}} \min _{u \in \mathcal{X}} p\left(x_{i-1} \mid X_0=u\right)}{1-\sum_{x_{i-1}} \min _{u \in \mathcal{X}} p\left(x_{i-1} \mid X_0=u\right)} = \frac{\alpha_{i} - \alpha_{i-1}}{1 - \alpha_{i-1}}.
\end{equation}
It is easy to see that if \eqref{eq:appendix-eq1} is established, then we immediately have
\begin{align*}
    p\left(Y_i= \ast \mid Y_{i-1} = \ast \right) = 1 - s_i = \frac{1 - \alpha_{i}}{1 - \alpha_{i-1}},
\end{align*}
which is our desired result.


To establish \eqref{eq:appendix-eq1}, we first inspect the relation between $p\left(x_i \mid x_0, y_{[i-1]}\right)$ and $p\left(x_{i+1} \mid x_0, y_{[i]}\right)$ for any $x_0, x_i, x_{i+1} \in \mathcal{X}$ as follows.
\begin{align*}
p\left(x_{i+1} \mid x_0, y_{[i]}\right) 
    & = \sum_{x_i} p\left(x_{i+1} \mid x_i\right) p\left(x_i \mid x_0, y_{[i-1]}, y_i\right) \\
    & = \sum_{x_i} p\left(x_{i+1} \mid x_i\right) \frac{p\left(x_i \mid x_0, y_{[i-1]}\right) p\left(y_i \mid x_0, x_i, y_{[i-1]}\right)}{\sum_{x_i} p\left(x_i \mid x_0, y_{[i-1]}\right) p\left(y_i \mid x_0, x_i, y_{[i-1]}\right)} \\
    & \utag{a}{=} \sum_{x_i} p\left(x_{i+1} \mid x_i\right) \frac{p\left(x_i \mid x_0, y_{[i-1]}\right)-\min _u p\left(x_i \mid X_0=u, y_{[i-1]}\right)}{1-\sum_{x_i} \min _u p\left(x_i \mid X_0=u, y_{[i-1]}\right)} \\
    &  \utag{b}{=}  \frac{\sum_{x_i} p\left(x_{i+1} \mid x_i\right)\left(p\left(x_i \mid x_0, y_{[i-1]}\right)-\min _u p\left(x_i \mid X_0=u, y_{[i-1]}\right)\right)}{1-\sum_{x_i} \min _u p\left(x_i \mid X_0=u, y_{[i-1]}\right)} \\
    & =\frac{\sum_{x_i} p\left(x_{i+1} \mid x_i\right)
    \left(p\left(x_i \mid x_0, y_{[i-1]}\right)-\min _u p\left(x_i \mid X_0=u, y_{[i-1]}\right)\right)}{1-s_i},
\end{align*}
where \uref{a} follows by plugging in the privacy mechanism \eqref{eq:qdef}, and \uref{b} follows because $1-\sum_{x_i} \min _u p\left(x_i \mid X_0=u, y_{[i-1]}\right)$ is a constant (does not depend on a particular value of $x_i$).

Given the above relation between $p\left(x_i \mid x_0, y_{[i-1]}\right)$ and $p\left(x_{i+1} \mid x_0, y_{[i]}\right)$, we claim that for any $x_0, x_i \in \mathcal{X}$,
\begin{align}
\label{eq:appendix-recur}
p\left(x_i \mid x_0, y_{[i-1]}\right) = \left(\prod_{k=1}^{i-1} \frac{1}{1-s_k}\right) \left(p\left(x_i \mid x_0\right)-\sum_{x_{i-1}} p\left(x_i \mid x_{i-1}\right) \min_{x_0} p\left(x_{i-1} \mid x_0\right)\right).
\end{align}

We prove it by inducting on $i$. For the base case $i=1$, we can easily verify that $\pb{x_1 \mid x_0, y_{0}}$ satisfies \eqref{eq:appendix-recur}.

Suppose that \eqref{eq:appendix-recur} is true for $p\left(x_i \mid x_0, y_{[i-1]}\right)$. Then we know that 
\begin{align*}
& p\left(x_i \mid x_0, y_{[i-1]}\right)-\min _u p\left(x_i \mid X_0=u, y_{[i-1]}\right)  = \left(\prod_{k=1}^{i-1} \frac{1}{1-s_k}\right) \left(p\left(x_i \mid x_0\right)- \min_{x_0}\pb{x_i \mid x_0} \right)
\end{align*}
because $\sum_{x_{i-1}} p\left(x_i \mid x_{i-1}\right) \min_{x_0} p\left(x_{i-1} \mid x_0\right)$ does not depend on a particular value of $x_0$.
Then we consider $p\left(x_{i+1} \mid x_0, y_{[i]}\right)$ as follows.
\begin{align*}
    \pb{x_{i+1} \mid x_0, y_{[i]}}
    & = \frac{\sum_{x_i} p\left(x_{i+1} \mid x_i\right)\left(p\left(x_i \mid x_0, y_{[i-1]}\right)-\min _u p\left(x_i \mid X_0=u, y_{[i-1]}\right)\right)}{1-s_i} \\
    & = \frac{\sum_{x_i} p\left(x_{i+1} \mid x_i\right)\left(
   \prod_{k=1}^{i-1} \frac{1}{1-s_k}\right) \left(p\left(x_i \mid x_0\right)- \min_{x_0}\pb{x_i \mid x_0} \right)
    }{1-s_i} \\
     & = 
    \left(\prod_{k=1}^{i} \frac{1}{1-s_k}\right) \left(\sum_{x_i} p\left(x_{i+1} \mid x_i\right) p\left(x_i \mid x_0\right)-\sum_{x_i} p\left(x_{i+1} \mid x_i\right)  \min_{x_0}\pb{x_i \mid x_0} 
    \right) \\
    & = 
    \left(\prod_{k=1}^{i} \frac{1}{1-s_k}\right) \left(
    p\left(x_{i+1} \mid x_0\right)
    -\sum_{x_i} p\left(x_{i+1} \mid x_i\right)  \min_{x_0}\pb{x_i \mid x_0} 
    \right),
\end{align*}
which justifies our claim. 

By taking minimization of $X_0$ and summation over $X_i$ on both sides of \eqref{eq:appendix-recur}, we obtain that
\begin{align*}
    \sum_{x_i} \min _{x_0} p\left(x_i \mid x_0, y_{[i-1]}\right)
    & = \sum_{x_i} \min _{x_0}\left(\left(\prod_{k=1}^{i-1} \frac{1}{1-s_k}\right)\left(p\left(x_i \mid x_0\right)-\sum_{x_{i-1}} p\left(x_i \mid x_{i-1}\right) \min _{x_0} p\left(x_{i-1} \mid x_0\right)\right)\right) \\
    & =    \left(\prod_{k=1}^{i-1} \frac{1}{1-s_k}\right)
    \left(
    \sum_{x_i} \min _{x_0} p\left(x_i \mid x_0\right)- \sum_{x_i} \min _{x_0} \sum_{x_{i-1}} p\left(x_i \mid x_{i-1}\right) \min _{x_0} p\left(x_{i-1} \mid x_0\right)
    \right)   \\
     & \utag{a}{=}    \left(\prod_{k=1}^{i-1} \frac{1}{1-s_k}\right)
    \left(
    \sum_{x_i} \min _{x_0} p\left(x_i \mid x_0\right)- \sum_{x_i}  \sum_{x_{i-1}} p\left(x_i \mid x_{i-1}\right) \min _{x_0} p\left(x_{i-1} \mid x_0\right)
    \right) \\
    & \utag{a}{=}    \left(\prod_{k=1}^{i-1} \frac{1}{1-s_k}\right)
    \left(
    \sum_{x_i} \min _{x_0} p\left(x_i \mid x_0\right)-  \sum_{x_{i-1}}  \min _{x_0} p\left(x_{i-1} \mid x_0\right)
    \right), 
\end{align*}
where \uref{a} follows because $\sum_{x_{i-1}} p\left(x_i \mid x_{i-1}\right) \min _{x_0} p\left(x_{i-1} \mid x_0\right)$
does not depend on a particular value of $x_0$.

Thus, the above equality can be summarized as
\begin{align*}
\sum_{x_i} \min _{x_0} p\left(x_i \mid x_0, y_{[i-1]}\right) & =\left(\prod_{k=1}^{i-1} \frac{1}{1-s_k}\right)\left(\sum_{x_i} \min _{x_0} p\left(x_i \mid x_0\right)-\sum_{x_i} \sum_{x_{i-1}} p\left(x_i \mid x_{i-1}\right) \min _{x_0} p\left(x_{i-1} \mid x_0\right)\right) \\
& =\left(\prod_{k=1}^{i-1} \frac{1}{1-s_k}\right)\left(\sum_{x_i} \min _{x_0} p\left(x_i \mid x_0\right)-\sum_{x_{i-1}} \min _{x_0} p\left(x_{i-1} \mid x_0\right)\right)
\end{align*}

Recalling the definition of $s_i$, we obtain that
\begin{align}
    s_i=\left(\prod_{k=1}^{i-1} \frac{1}{1-s_k}\right)\left(\sum_{x_i} \min _{x_0} p\left(x_i \mid x_0\right)-\sum_{x_{i-1}} \min _{x_0} p\left(x_{i-1} \mid x_0\right)\right) 
    = \left(\prod_{k=1}^{i-1} \frac{1}{1-s_k}\right)\left(
    \alpha_i -\alpha_{i-1} \right)  . \label{eq13}
\end{align}

Finally, we prove the desired result $s_i  = \frac{\alpha_{i} - \alpha_{i-1}}{1 - \alpha_{i-1}}$
by induction on $i$. For the base case $i = 1$, we can see that 
$s_1 = \alpha_1 - \alpha_0 =  \frac{\alpha_{1} - \alpha_{0}}{1 - \alpha_{0}}$
by the fact that $\alpha_0 = 0$. Suppose that it is true for $i - 1$. Then, we have for $i$
\begin{align*}
    s_i 
    & = \left(\prod_{k=1}^{i-1} \frac{1}{1-s_k}\right)\left(
    \alpha_i -\alpha_{i-1} \right)  \\
    & =  \left(
    \prod_{k=1}^{i-1} 
    \frac{1 - \alpha_{k-1}}{1- \alpha_{k} } 
    \right)
    \left( \alpha_i -\alpha_{i-1} \right) \\
    & =  \frac{(1-\alpha_0)(\alpha_i -\alpha_{i-1})}{1- \alpha_{i-1} } \\
    & = \frac{\alpha_i -\alpha_{i-1}}{1- \alpha_{i-1} },
\end{align*}
as desired since $\alpha_0 = 0$ in the last step.


\newpage

    \end{document}